\documentclass[letterpaper,twocolumn,10pt]{article}
\usepackage{graphicx}
\usepackage{amsmath,amssymb,amsthm,mathtools}
\usepackage{xspace}
\usepackage{placeins}
\usepackage{booktabs}
\usepackage{tikz} 
\usetikzlibrary{decorations.pathreplacing}
\usepackage[utf8]{inputenc}
\usepackage{algorithm}
\usepackage[noend]{algpseudocode}
\usepackage{times} 
\IfFileExists{usenix-2020-09.sty}{\usepackage{usenix-2020-09}}{}

\usepackage{etoolbox}
\newbool{publicversion}
\newbool{tr}
\newbool{optanon}
\newbool{compress}
\newbool{usenix}
\newbool{sigplan}
\newbool{sigconf}
\newbool{lipcs}
\newbool{sigplan-danger}

\setbool{tr}{true}
\setbool{publicversion}{false}
\setbool{compress}{true}
\setbool{optanon}{true}
\setbool{usenix}{true}
\setbool{sigplan}{false}
\setbool{sigconf}{false}
\setbool{lipcs}{false}
\setbool{sigplan-danger}{false}

\ifbool{sigplan-danger}{
\usepackage{array}
\usepackage{appendix}
\usepackage{booktabs} 
\usepackage{blindtext}
\usepackage[font={it,small}]{caption,subcaption}
\usepackage{cancel}
\usepackage{changepage}
\usepackage[xcolor]{changebar}
\usepackage[english]{babel}
\usepackage{epsfig}
\usepackage{endnotes}
\usepackage{lipsum} 
\usepackage{makecell} 
\usepackage{mdframed}
\usepackage{parskip}
\usepackage{tikz}
\usetikzlibrary{positioning}
\usetikzlibrary{patterns}
\usepackage[normalem]{ulem}
\usepackage{xurl}
}
\usepackage{algorithm}
\usepackage{mathrsfs}
\usepackage{amsthm}
\usepackage{amsmath}
\usepackage{balance}
\usepackage{color}
\usepackage{comment} 
\usepackage[shortlabels]{enumitem} 
\usepackage[export]{adjustbox}
\usepackage{etoolbox}
\usepackage{listings}
\usepackage{longfbox}
\usepackage{geometry}
\usepackage{hyperref} 
\usepackage[utf8]{inputenc}
\usepackage{multirow}
\usepackage{pifont}
\usepackage[colorinlistoftodos]{todonotes} 
\usepackage[compact]{titlesec}
\usepackage{setspace}
\usepackage{textcomp}
\usepackage{url}
\usepackage{xspace}
\usepackage{xcolor}
\usepackage{soul}

\ifbool{usenix}{
    \usepackage[numbers,sort&compress,square]{natbib}
    \usepackage[normalem]{ulem} 

    \setlength{\textheight}{9.0in}
    \setlength{\columnsep}{0.33in}
    \setlength{\textwidth}{7.00in}
    
    \setlength{\topmargin}{0.0in}
    
    \setlength{\headheight}{0.0in}
    
    \setlength{\headsep}{0.0in}
    
    \addtolength{\oddsidemargin}{-0.25in}
    \addtolength{\evensidemargin}{-0.25in}

}{
    \ifbool{sigplan} {
        \settopmatter{printacmref=false} 
        \settopmatter{printfolios=true} 
        \renewcommand\footnotetextcopyrightpermission[1]{} 
      
        \pagestyle{plain} 
    }{}
}

\definecolor{gold}{RGB}{218,165,32}
\definecolor{lightblue}{rgb}{0.6,1.0,0.6}
\definecolor{codegreen}{rgb}{0,0.6,0}
\definecolor{gray}{rgb}{0.5,0.5,0.5}
\definecolor{mauve}{rgb}{0.58,0,0.82}

\ifbool{publicversion}{
  \newcommand{\grumbler}[2]{}
  \newcommand{\grumblermargin}[2]{}
  
  \newcommand{\pagelimit}[1]{}
  \newcommand{\todos}[1]{}
} 
{
  \newcommand{\grumbler}[2]{{{\bf #1:} #2}}
  
 \newcommand{\grumblermargin}[3]{\todo[size=\tiny,color=#3, nolist]{{\bf #1:} #2}}
  \newcommand{\pagelimit}[1]{{\hl{ #1 pages }}}

  \newcommand{\todos}[1]{\todo[size=\small, inline]{#1}}
}

\algdef{SE}[EVENT]{Event}{EndEvent}[1]{\textbf{upon}\ #1\ \algorithmicdo}

\ifbool{compress}{
  \setlist[itemize]{leftmargin=*}
  \setlist[enumerate, 1]{1.} 
  \setlist{itemsep=0pt,parsep=2pt}             
  \usepackage[subtle,wordspacing=normal]{savetrees} 
  \setlength{\parskip}{1pt plus 1pt minus 1pt} 
  \setlength{\parindent}{5pt} 
  \setlength\abovedisplayskip{2pt plus 2pt minus 2pt}
  \setlength\belowdisplayskip{2pt plus 2pt minus 2pt}
  \setlength{\abovecaptionskip}{1pt plus 2pt minus 2pt} 
  \setlength{\belowcaptionskip}{1pt plus 2pt minus 2pt} 
  \setlength\floatsep{2pt plus  2pt minus 2pt}
  \setlength\textfloatsep{3pt plus 2pt minus 2pt}
  \titlespacing{\section}{0pt}{6pt plus 2pt minus 2pt}{3pt plus 1pt minus 1pt} 
  \titlespacing{\subsection}{0pt}{1ex}{0.33ex}  
  \titlespacing{\subsubsection}{0pt}{0.6ex}{1ex}
} 
{
\setlist{itemsep=2pt,parsep=1pt}             
}

 \ifbool{optanon}%
   {}%
   {}
 
 \ifbool{tr}{
    
}
{
    
}

\newtheoremstyle{indented}
{1pt}   
{1pt}   
{\itshape}  
{1pt}       
{\bfseries} 
{}         
{2pt plus 1pt minus 1pt} 
{}          

\newtheoremstyle{noindent}
{2pt}   
{2pt}   
{\itshape}  
{0pt}       
{\bfseries} 
{}         
{2pt plus 1pt minus 1pt} 
{}          

\newtheoremstyle{definition-compressed}
{2pt}   
{2pt}   
{\itshape}  
{0pt}       
{\bfseries} 
{}         
{2pt plus 1pt minus 1pt} 
{}          

\theoremstyle{plain}
\newtheorem{lemma}{Lemma}
\newtheorem{theorem}{Theorem}
\theoremstyle{noindent}
\newtheorem{lemma-noindent}{Lemma}
\newtheorem{theorem-noindent}{Theorem}
\theoremstyle{definition-compressed}

\fboxset{rounded}

\addcontentsline{toc}{chapter}{Appendix}

\lstdefinelanguage{thrift}
    {morekeywords={struct,Distribution},
    sensitive=false,
    morecomment=[l]{//},
    morecomment=[s]{/*}{*/},
    morestring=[b]",
}

\lstdefinestyle{placement}{
  float=tp,
  floatplacement=tbp,
}

\tikzset{
    position/.style args={#1:#2 from #3}{
        at=(#3.#1), anchor=#1+180, shift=(#1:#2)
    }
}

\usepackage[dvipsnames]{xcolor} 
\usepackage{amsmath} 
\usepackage{todonotes}
\usepackage[most]{tcolorbox}
\usepackage[normalem]{ulem} 

\newif\ifshowrev
\showrevtrue                       

\newtcolorbox{arnewremarkbox}{
    colback=violet!5!white,    
    colframe=violet!75!black,  
    fonttitle=\bfseries,
    title=A,
    fontupper=\small\itshape,
    before upper={$\blacktriangleright$\space},
    after upper={\space$\triangleleft$},
}

\newcommand{\consensus}
{Hermes\xspace}

\newcommand{\SC}{\mathrm{SC}}

\title{\textbf{\consensus: Low Tail-Latency Via Prefix Consensus}}
\author{
{\rm Alejandro Ranchal-Pedrosa}\\
Sei Labs
\and
{\rm Dakai Kang}\\
Sei Labs
\and
{\rm Neil Giridharan}\\
University of California, Berkeley
\and
{\rm Mohammad Sadoghi}\\
University of California, Davis
\and
{\rm Dahlia Malkhi}\\
UC Santa Barbara
\and
{\rm Ben Marsh}\\
Sei Labs, University of Portsmouth}

\date{}
\begin{document}
\newtheorem{corollary}[theorem]{Corollary}     

\maketitle
\begin{abstract}
Leader-based BFT protocols finalize through their leaders: a view whose
leader is crashed or slow finalizes nothing, and the timeout that ends it
admits no good setting. A conservative timeout turns every crashed leader
into a long stall; an aggressive one voids the views of leaders that are
merely slow. Either way the expired view is wasted, and this trade-off, not
the good case, governs tail latency.

Hermes makes expired views finalize. Hermes is a two-round rotating-leader
protocol for $n=5f+1$ processes under partial synchrony built on prefix
consensus: votes carry values ordered by a prefix relation, and quorums
require comparability rather than equality. Every process broadcasts a
justified proposal at view start and casts a single vote, for the leader's
proposal upon delivery or for a fallback proposal at the timeout; there are
no nullify votes. A timely honest leader finalizes its full proposal in
$2\delta$. Otherwise, any $n-f$ votes, which need not match, finalize the
heaviest common prefix. We instantiate Hermes as a finality gadget over an
available chain and over Autobahn-style multi-lane dissemination, where
parent-relative delta tipcuts with explicit skips keep independent proposals
comparable and sender-indexed erasure-coded amplification keeps complexity
at $O(n^2\hat{m}+\lambda n^3)$ bits and $O(n^2)$ messages per view, where $\hat{m}$ is the number of concurrent dissemination lanes in Autobahn and $\lambda$ the security parameter.
\end{abstract}
\section{Introduction}
\label{sec:intro}

Leader-based Byzantine fault tolerant (BFT) protocols order requests in
views, each driven by a designated leader. A long line of work has driven
the view that goes well to its floor: with $n=5f+1$ processes, up to $f$ of
them Byzantine, a timely honest leader finalizes in two message delays
($2\delta$)~\cite{martin2006fast, dutta2005bestcase}, and a recent wave of protocols builds at exactly
this design point~\cite{chou2025minimmit,kniep2025alpenglow,shoup2025kudzu,franca2025chonkybft,shrestha2026hydrangea}. The
good case is, at this point, essentially closed. The tail is not.

Latency in a rotating-leader protocol is a distribution over views. A view
whose leader is crashed, partitioned, or slow finalizes nothing: processes
wait out a timeout, vote to nullify or to change views, and hand the data to
the next leader. The timeout mediates the resulting trade-off, and no
setting serves both sides. Set conservatively, it tolerates slow leaders but
converts every crashed leader into a full timeout of dead time. Set
aggressively, it skips crashed leaders quickly but voids the views of
leaders that are merely slow, a large percentile in geo-distributed
deployments where round-robin leadership visits every latency profile.
Either choice surfaces as tail latency: a request that lands in a bad view
pays the timeout, the view change, and the next leader's round trip before
it can finalize. Recent work attacks the problem by racing the leader path
against recovery~\cite{giridharan2026ambulance}; we attack it with \consensus by making the timeout path
itself finalize.
\begin{table*}[t]
\centering
\small
\begin{tabular}{@{}l c c l l c@{}}
\toprule
Protocol & $n$ & Good case & View with failed leader & Votes range over & Proposers per view \\
\midrule
Minimmit~\cite{chou2025minimmit} & $5f{+}1$ & $2\delta$ & nullified & single value, exact match & leader \\
Raptr~\cite{tonkikh2025raptr} & $3f{+}1$ & $3\delta$ & view change, no output & prefixes of the leader's proposal & leader \\
Floor-IT~\cite{abraham2026floorit}\,$^{\S}$ & $5f{-}1$ & $2\delta$ & view change, no output & single value, exact match & leader \\
Multimmit~\cite{lewispye2026multimmit}\,$^{\dagger}$ & $5f{+}1$ & $2\delta$ & nullified & leader's + per-chain extensions & leader \\
Prefix Consensus~\cite{xiang2026prefix}\,$^{\dagger}$ & $3f{+}1$\,$^{\ddagger}$ & $4\delta$ (amort.) & leaderless & vectors, per-coordinate exact match & all \\
\textbf{Hermes} & $5f{+}1$ & $2\delta$ & \textbf{finalizes HCP at $2\Delta{+}\delta$} & independent proposals, ordered by $\preceq$ & all \\
\bottomrule
\end{tabular}
\caption{Hermes against the closest designs. Good case is post-GST
finalization latency under a timely honest leader; for the leaderless
protocol it is their amortized commit latency in failure-free executions
with synchronized slot starts, in which the slot itself finalizes in
$8\delta$. $\dagger$: concurrent and independent work. $\S$:
signature-free; $5f{-}1$ is the tight threshold in that model.
$\ddagger$: a two-round construction at $n\geq5f{+}1$ is given for the
one-shot primitive; their SMR builds only on the three-round primitive,
tight at $3f{+}1$. Hermes is the only protocol in which the view of a
failed leader itself finalizes.}
\label{tab:intro-comparison}
\end{table*}
The waste in a voided view is larger than it looks. Deployments increasingly
separate dissemination from ordering: DAG mempools and lane-based designs
spread transactions before consensus sequences
them~\cite{danezis2022narwhal,giridharan2024autobahn}, and finality gadgets order blocks that an
available chain has already propagated. 
In these settings a proposal is not an arbitrary value but a prefix of
the stream the dissemination layer has produced, and the honest proposals
within a view differ only in a suffix of network jitter. Yet existing consensus protocols
count votes only when they match. Votes for two proposals that agree
on all but their last few blocks are therefore treated as disagreement and discarded at timeout.

Hermes is a two-round protocol for $n=5f+1$ processes under partial
synchrony whose views finalize with or without their leader. In the first
round, every process broadcasts a justified proposal; the round-robin
leader's is one among $n$. In the second, every process casts a single vote:
for the leader's proposal upon delivering it, or, once a $2\Delta$ timeout
elapses and $n-f$ valid proposals are held, for a fallback proposal among
them. Justifications are assembled from the votes themselves, so a view has
neither extra rounds nor nullifications, with just $2f+1$ distinct, but pairwise comparable votes sufficing to advance the view.

The timeout dilemma dissolves in \consensus. 
An aggressive timeout costs slow leaders only the suffix on which the votes disagree, not the view itself: a set of votes that do not all match, which previous consensus protocols can only use to skip the view, instead finalizes in \consensus the common prefix those votes share, at $2\Delta+\delta$. Skipping a slow leader quickly therefore no longer means wasting a view.

Sections~\ref{sec:model} and~\ref{sec:protocol} present Hermes over an
abstract value space: values ordered by a prefix relation $\preceq$ under
which the prefixes of any value are finitely many and pairwise comparable.
Section~\ref{sec:prefixconsensus} instantiates two different applications depending on the value space, leaving the view
structure untouched. The single-lane instantiation as a finality gadget over
an available chain: values are chain tips and $\preceq$ is ancestry. The multi-lane instantiation orders Autobahn-style lanes while addressing three challenges: 
\begin{itemize}
    \item \emph{Comparability:} deterministic interleaving and parent-relative delta tipcuts with explicit skips keep proposals comparable despite heterogeneous lane speeds. 
    \item \emph{Compactness:} votes and parent links use identifiers, yielding $O(\hat m)$-word proposals and $O(\lambda n^3)$-bit justification forwarding, rather than $O(n^3\hat m)$ bits for full values (quartic when $\hat m=\Theta(n)$). 
    \item \emph{Availability:} because prefix extraction examines the divergent values behind (4f+1) votes, processes erasure-code and amplify available proposals and act only on amplified values. Every justified value is therefore reconstructible by all processes, at $O(n^2\hat m+\lambda n^3)$ bits and $O(n^2)$ messages per view, with no fragment traffic on the good path.
\end{itemize}

Prefix votes are not new, their role here is. Raptr~\cite{tonkikh2025raptr} votes
prefixes of a single leader proposal at $n=3f+1$ to tolerate partial data
availability; Hermes intersects independently proposed histories after the
leader path fails, at the replication two-delay finality already pays. Concurrently, Xiang
et al.~\cite{xiang2026prefix} formalize a leaderless prefix consensus
abstraction through low and high output vectors; Hermes keeps the leader
fast path and extracts one nested history per view. Minimmit~\cite{chou2025minimmit}
is the closest single-value protocol: Hermes keeps its quorum sizes and its
one vote per process per view, and changes two things, every process
proposes and quorums require comparability instead of equality. A view
Minimmit nullifies, Hermes finalizes. Also concurrently, Multimmit~\cite{lewispye2026multimmit} extends
Minimmit across producer chains yet still voids a view whose leader block
cannot be certified.

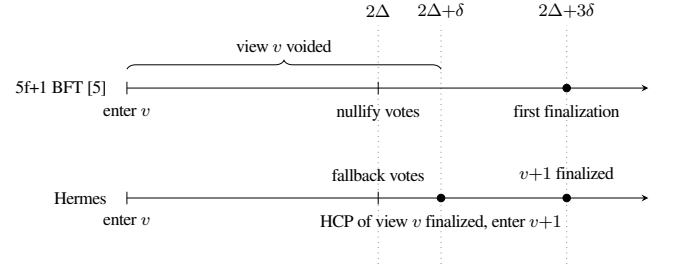
\begin{figure}[t]
\centering
\resizebox{\columnwidth}{!}{%
\begin{tikzpicture}[>=stealth, font=\footnotesize]
  \foreach \x/\lab in {4/{$2\Delta$}, 5/{$2\Delta{+}\delta$}, 7/{$2\Delta{+}3\delta$}}{
    \draw[dotted, gray] (\x,-0.8) -- (\x,3.0);
    \node[above] at (\x,3.0) {\lab};
  }
  \node[anchor=east] at (-0.2,2.0) {5f+1 BFT~\cite{chou2025minimmit}};
  \draw[->] (0,2.0) -- (8.3,2.0);
  \draw (0,1.9) -- (0,2.1); \node[below] at (0,1.85) {enter $v$};
  \draw (4,1.9) -- (4,2.1); \node[below] at (4,1.85) {nullify votes};
  \filldraw (7,2.0) circle (1.7pt); \node[below] at (7,1.85) {first finalization};
  \draw[decorate,decoration={brace,amplitude=4pt}] (0,2.3) -- (5,2.3)
       node[midway,above=4pt] {view $v$ voided};
  \node[anchor=east] at (-0.2,0.25) {Hermes};
  \draw[->] (0,0.25) -- (8.3,0.25);
  \draw (0,0.15) -- (0,0.35); \node[below] at (0,0.1) {enter $v$};
  \draw (4,0.15) -- (4,0.35); \node[above] at (4,0.4) {fallback votes};
  \filldraw (5,0.25) circle (1.7pt);
  \node[below] at (5,0.1) {HCP of view $v$ finalized, enter $v{+}1$};
  \filldraw (7,0.25) circle (1.7pt);
  \node[above] at (7,0.4) {$v{+}1$ finalized};
\end{tikzpicture}}
\caption{One view with a slow leader, both protocols with the same
$2\Delta$ timeout; time in units of $\delta$, drawn with $\Delta=2\delta$.
The single-value protocol nullifies the view and first finalizes when the
next leader re-proposes, at $2\Delta{+}3\delta$. In Hermes the same timeout
produces fallback votes, any $n{-}f$ of which finalize the heaviest common
prefix of view $v$ at $2\Delta{+}\delta$; those votes justify view
$v{+}1$, whose leader finalizes at $2\Delta{+}3\delta$ in both protocols.
Skipping aggressively loses nothing in the common case.}
\label{fig:dilemma}
\end{figure}
This paper contributes:
\begin{enumerate}
\item Hermes, a two-round rotating-leader protocol at $n=5f+1$ with no
nullification path: $2\delta$ finalization under a timely honest leader, and
finalization of the heaviest common prefix of any $n-f$ votes at
$2\Delta+\delta$ otherwise, with proofs of safety and liveness over an
abstract prefix value space (Sections \ref{sec:model}, \ref{sec:protocol}, \ref{sec:proofs}).
\item A single-lane instantiation: a finality gadget over a dynamically
available chain (Section~\ref{subsec:singlelane}).
\item A multi-lane instantiation over Autobahn-style lanes: deterministic
interleaving, parent-relative delta tipcuts with explicit skip
certification, and sender-indexed erasure-coded amplification, at
$O(n^2\hat m+\lambda n^3)$ bits and $O(n^2)$ messages per view
(Section~\ref{subsec:multilane}).
\item An implementation and an evaluation (Section~\ref{sec:eval}).
\end{enumerate}

\section{Model}
\label{sec:model}
We assume a BFT model with $n=5f+1$ processes, where up to $f$ can be Byzantine. All other processes are honest and follow the protocol.

\paragraph{Network.}
All channels are authenticated and pairwise. We assume partial synchrony: before an unknown Global Stabilization Time (GST), message delivery may be arbitrarily delayed; after GST, every message sent by an honest process arrives within a known bound $\Delta$. We write $\delta \le \Delta$ for the actual (unknown) network delay.

\paragraph{Adversary.}
The adversary controls the $f$ Byzantine processes with arbitrary malicious behaviors such as learning their entire state and choosing their messages. The adversary also schedules message delivery subject to the partial synchrony bound. It is computationally bounded and cannot break standard cryptographic primitives.

\paragraph{Cryptographic assumptions and parameters.}
We assume digital signatures: every process can sign messages, and any process can verify any other process' signatures. A collision-resistant hash function $H$ backs every reference to a value that is not the value itself: block references within chains, the parent identifiers of Section~\ref{subsubsec:threeindex}, and the coded identifiers of Section~\ref{sec:amplification}. 

\paragraph{Values.} We write $c$ for the size of a proposal. Proposals carry \emph{values} from a space equipped with a prefix relation $\preceq$: $B \preceq B'$ reads $B$ is a prefix of $B'$, and two values are \emph{pairwise comparable} if one is a prefix of the other. We require two properties of the space, both immediate when values are tips of hash chains and $\preceq$ is the ancestor relation: the prefixes of any single value are pairwise comparable and finitely many, and a distinguished \emph{genesis} value is a prefix of every value. 
Section~\ref{sec:prefixconsensus} instantiates the value space for a single chain (e.g. as a finality gadget for an available chain) and for multi-lane dissemination (e.g. consensus preceded by Autobahn DA~\cite{giridharan2024autobahn}), changing the representation of values and of the objects below.
The relation $\preceq$ is a partial order; in particular it is reflexive, and \emph{extends} includes equality throughout.

\paragraph{Heaviest Common Prefix and Selected Continuation.}
\label{subsubsec:hcp-single}
For a value $x$ and a multiset $\mathcal{S}$ of voted values, define the support of $x$ inside $\mathcal{S}$ as
\[
    \mathrm{supp}_{\mathcal{S}}(x)
    =
    \bigl|\{B \in \mathcal{S} : x \preceq B\}\bigr|,
\]
counted with multiplicity: two votes endorsing the same value contribute two.
For a multiset $\mathcal{S}$ of values endorsed by $4f+1$ votes from distinct processes, the \emph{heaviest common prefix} and the \emph{selected continuation} are
\[
    \mathsf{HCP}(\mathcal{S})
    :=
    \max_{\preceq}
    \bigl\{x : \mathrm{supp}_{\mathcal{S}}(x)=|\mathcal{S}|\bigr\},
\]
\[\mathrm{SC}(\mathcal{S})
    :=
    \max_{\preceq}
    \bigl\{x : \mathrm{supp}_{\mathcal{S}}(x) \ge 2f+1\bigr\}.
\]
Equivalently, $\mathsf{HCP}(\mathcal{S})$ is the greatest common prefix of the full vote set, and $\mathrm{SC}(\mathcal{S})$ is the $\preceq$-greatest among the greatest common prefixes of the $(2f{+}1)$-vote subsets of $\mathcal{S}$.

Both maxima are well defined, with no tie-breaking involved. Any two values with support at least $2f+1$ in $\mathcal{S}$ are comparable: their supports total at least $4f+2 > |\mathcal{S}|$, so some value of $\mathcal{S}$ lies in both. The set of values with support at least $2f+1$ is therefore totally ordered by $\preceq$. Full support is the special case $4f+1 \ge 2f+1$, so the same holds for $\mathsf{HCP}(\mathcal{S})$.

Intuitively, $\mathsf{HCP}(\mathcal{S})$ is the prefix that is safe to finalize from the $4f+1$ votes, while $\mathsf{SC}(\mathcal{S})$ is the continuation that later proposals must extend in subsequent views.

\section{\consensus}
\label{sec:protocol}
\consensus runs in views, each with a publicly known round-robin leader $L_v = v \bmod n$. The design goal is that every view finalizes a new prefix irrespective of slow leaders: the leader's full value when the leader is timely, and the heaviest common prefix of the votes cast in the view otherwise. The view timer marks which of the two a view finalizes, diminishing the usual tension between a long timeout that tolerates some slow leaders without skipping (but takes longer to skip views by even slower or crashed leaders) and a short one that creates more skipped views (but replaces them more quickly). Thanks to finalizing the heaviest common prefix, it is possible to target an aggressive timeout at the expense of the tail of slow leaders, while at the same time reducing the probability that a view produces no new finalized decision to rare edge cases.

\paragraph{Protocol overview.}
A view has two rounds. In Round~1, every process $p_i$ broadcasts its own justified proposal $B_i$; the leader's proposal is one among $n$. This all-to-all burst supplies the material from which a view finalizes when the leader path does not complete. In Round~2, every process casts a single vote: for the leader's proposal upon delivering it, or, once the $2\Delta$ timeout has elapsed and $n-f$ valid proposals have been received, for a fallback proposal among them. There are no nullify votes.

A process $p_i$ may enter view $v{+}1$ only if it holds a \emph{justification} $J^{(v)}$ from view $v$ that is one of:
\begin{enumerate}[(i)]
  \item \textbf{M-notarization:} $2f{+}1$ pairwise comparable votes from distinct processes in view $v$; or 
  \item \textbf{L-notarization:} \emph{any} set of $n{-}f=4f{+}1$ votes from distinct processes in view $v$.
\end{enumerate}
An L-notarization moreover finalizes the heaviest common prefix of its votes. 

A justification from view $v$ also determines what view $v{+}1$ may build on: a proposal for view $v{+}1$ is \emph{justified} by $J^{(v)}$ only if it extends
\begin{enumerate}[(i)]
  \item the \emph{M-notarized tip}, if $J^{(v)}$ is an M-notarization: the $\preceq$-greatest of its $2f{+}1$ pairwise comparable votes, which extends all of them; or
  \item the \emph{selected continuation} $\mathrm{SC}(\mathcal{S}_{J^{(v)}})$ of the multiset $\mathcal{S}_{J^{(v)}}$ of values its votes endorse.
\end{enumerate}
Since justifications are assembled from the votes themselves, a view has no more rounds or nullifications. Algorithm~\ref{alg:hermes-psync} gives the full instructions for $p_i$; we walk through it next.
\begin{algorithm*}[h]
\caption{\consensus instructions for $p_i$; Algorithms~\ref{alg:singlelane-hooks} and~\ref{alg:derivesequence} instantiate the auxiliary functions.}
\label{alg:hermes-psync}
\begin{algorithmic}[1]
\State send the first verified L-notarization of each view to all processors \Comment{Forward justifications} \label{line:fwd-lnot}
\State send the first verified M-notarization of each view to all processors \label{line:fwd-mnot}
\State $L_v \gets v \bmod n$ \Comment{Deterministic round-robin leader} \label{line:leader}
\If{proposed $=$ false} \label{line:propose-guard}
    \State $(B_{v-1}, J^{(v)}) \gets \Call{SelectParent}{S, v}$ \label{line:selectparent}
    \State $B \gets \Call{GetChild}{B_{v-1}}$ \label{line:getchild}
    \State send $J^{(v)}$ to all processors if not previously sent \Comment{Forward selected justification}
    \State \textbf{send} (\textsc{Propose}, $B$, $v$) to all processors \label{line:send-propose}
    \State proposed $\gets$ true \label{line:set-proposed}
    \State $\text{timer}_v \gets \text{now}()$ \Comment{Start view timer} \label{line:start-timer}
\EndIf
\If{$S$ contains a valid proposal $B_{L_v}$ for view $v$ from leader $p_{L_v}$} \label{line:leaderprop-guard}
    \If{voted$_v$$=$ false} ; voted$_v$$\gets$ true \label{line:voteleader-guard}
\State \textbf{send} (\textsc{Vote}, $v$, $B_{L_v}$) to all processors \Comment{Vote for leader} \label{line:send-voteleader}
\EndIf
\EndIf
\If{$\text{now}() - \text{timer}_v \geq 2\Delta$ \textbf{and} $S$ contains $C = n - f$ valid proposals for view $v$} \label{line:timeout-guard}
    \If{voted$_v$$=$ false} ; voted$_v$$\gets$ true \label{line:votefallback-guard}
\State \textbf{send} (\textsc{Vote}, $v$, \Call{SelectProposal}{$C$}) to all processors \Comment{Vote fallback after timeout} \label{line:send-votefallback}    \EndIf
\EndIf
\If{$S$ contains an M-notarization $M$ for tip $B$ with $B.\text{view} = v'$} \label{line:mnot-guard}
    \If{voted$_{v'}$ $=$ false} ; voted$_{v'}$ $\gets$ true \label{line:votemnot-guard}
        \State \textbf{send} (\textsc{Vote}, $v'$, $B$) to all processors \Comment{Catch-up vote when $v' > v$, straggler vote when $v' < v$} \label{line:send-votemnot}
    \EndIf
    \If{$v' \geq v$} \label{line:mnot-jump-guard}
        \State $v \gets v' + 1$, proposed $\gets$ false \Comment{Jump} \label{line:advance-mnot}
    \EndIf
\EndIf
\If{$S$ contains a new L-notarization $L$ for view $v'$ with heaviest common prefix $B$} \label{line:lnot-guard}
    \State $H \gets \Call{GetFinalizedHead}{\null}$ \label{line:gethead}
    \If{$H \preceq B$} \label{line:head-check}
        \State \Call{AppendToHead}{$B$} \Comment{Finalize, also when $v' < v$} \label{line:append}
    \EndIf
    \If{$v' \geq v$} \label{line:lnot-jump-guard}
        \If{voted$_{v'}$ $=$ false} ; voted$_{v'}$ $\gets$ true \label{line:votelnot-guard}
            \State \textbf{send} (\textsc{Vote}, $v'$, $\mathrm{SC}(\mathcal{S}_L)$) to all processors \Comment{Jump vote} \label{line:send-votelnot}
        \EndIf
        \State $v \gets v' + 1$, proposed $\gets$ false \Comment{Jump} \label{line:advance-lnot}
    \EndIf
\EndIf
\end{algorithmic}
\end{algorithm*}

\paragraph{Proposing (lines \ref{line:fwd-lnot} to \ref{line:start-timer}).}
Lines \ref{line:fwd-lnot} and \ref{line:fwd-mnot} forward the first L-notarization and the first
M-notarization that $p_i$ verifies for each view. Forwarding brings late processes up
to date and lets justification remain implicit: a proposal need not
attach the votes that justify it. Line \ref{line:selectparent} selects, from the local state $S$, a valid justification from the previous view together with the parent value it induces, and line \ref{line:getchild} forms the proposal as a child extending that parent. 

The justification $J^{(v)}$ that $p_i$ gathers restricts the new proposal $B_i$ to extend the value $J^{(v)}$ justifies, as defined above. Unlike the fallback vote, this restriction carries the safety of the protocol: it is what guarantees that every valid proposal of view $v{+}1$ extends every prefix that view $v$ can finalize.

For readability, Algorithm~\ref{alg:hermes-psync} writes the endorsed value in each \textsc{Vote}; on the wire, a vote carries only an $\mathcal{O}(\lambda)$-bit collision-resistant identifier, whose resolution is specified by the chosen value-space instantiation in Section~\ref{sec:prefixconsensus}.

\paragraph{Voting (lines \ref{line:leaderprop-guard} to \ref{line:send-votefallback}).}
A process votes for the leader's proposal upon delivering it (lines \ref{line:leaderprop-guard} to \ref{line:send-voteleader}). If instead the $2\Delta$ timer expires and $n-f$ valid proposals $C$ have been received, it votes for the fallback proposal $\Call{SelectProposal}{C}$ (lines \ref{line:timeout-guard} to \ref{line:send-votefallback}). The choice of fallback is a progress policy, not a safety condition: safety holds for any choice, and the choice only affects how much a view finalizes when the leader path fails. We prescribe the following choice: each view has a deterministic fallback leader; a process votes for the fallback leader's proposal if delivered timely, and for its own proposal otherwise. After GST, in a view with an honest leader, \consensus obtains the lower-bound on consensus latency of $2\delta$. Before GST, or whenever the leader is faulty or slow, the fallback votes still assemble L-notarizations, so these views keep finalizing heaviest common prefixes. 

\paragraph{Advancing views, catch-up, and finalizing (lines \ref{line:mnot-guard} to \ref{line:advance-lnot}).}
Any justification for a view $v' \geq v$ moves the process to view $v'+1$, so a process behind synchronizes in a single delivery of the highest justification it is missing, which lines \ref{line:fwd-lnot} and \ref{line:fwd-mnot} keep in flight. An M-notarization for a tip $B$ of any view $v'$ first triggers a view-$v'$ vote for $B$ at a process that has not voted in $v'$ (lines \ref{line:votemnot-guard} and \ref{line:send-votemnot}): straggler votes reinforce the notarized value. If $v' \geq v$, the process then jumps to the newer view. A new L-notarization $L$ for any view $v'$ finalizes its heaviest common prefix $B$ whenever $B$ extends the local finalized head (lines \ref{line:gethead} to \ref{line:append}), regardless of the current view. If moreover $v' \geq v$, a process that has not voted in $v'$ casts its view-$v'$ vote for the selected continuation $\mathrm{SC}(\mathcal{S}_L)$, the value proposals justified by $L$ must extend (lines \ref{line:votelnot-guard} and \ref{line:send-votelnot}), and starts the next view (line \ref{line:advance-lnot}). The flags voted$_{v'}$ are initialized to false for every view and never reset.

It is safe for the finalization rule to trigger for any past view $v$ regardless of the current view $v'\geq v$. Preceded by a dissemination layer, the extracted prefix extends the finalized state whenever the layer has delivered new data, so a view is unproductive only when there was nothing new to order or in unlucky cases where there was a split in two proposals from the dissemination layer that share no new prefix.

We specify constructions of \consensus preceded by an available chain and by Autobahn DA in Section~\ref{sec:prefixconsensus}, which instantiates the value space, the parent selection, and the fallback choice of this section, leaving the view structure of Algorithm~\ref{alg:hermes-psync} unchanged.

\FloatBarrier
\section{Prefix Consensus}
\label{sec:prefixconsensus}
The main benefit of \consensus is the relaxation of the trade-off between responsiveness (how fast a view changes) and stability (the number of view changes that produce progress) in cases where it is likely that proposals share a common prefix: if a view changes so quickly that a slow leader cannot propose timely, the backup leader can get proposals finalized; if instead a proposer is fast enough to reach only certain voters timely, introducing divergence, the finalization of the common prefix likely produces a nonempty decision.

This section instantiates a number of value spaces that can be used by the main \consensus protocol of Section~\ref{sec:protocol}.
The propose/vote/view-change flow remains exactly as specified there.
We specify here how values are represented, how \Call{SelectParent}{} and \Call{SelectProposal}{} are instantiated, and how a process extracts a finalized prefix from any $4f+1$ votes in the same view.

Whenever a deterministic choice among $\preceq$-incomparable values is needed (i.e. by equivocation or by different subsets of $4f+1$ votes), we fix a total order $\triangleleft$ on canonically serialized values: concretely, lexicographic order on the hash of the canonical serialization. Any selection mechanism is valid to break ties, and it does not need to be enforced by the protocol.

Throughout, a superscript names the view and a subscript the sender: $B^{(v)}_j$ is the value proposed by $p_j$ in view $v$, and likewise $T^{(v)}_j$ for the tipcuts of Section~\ref{subsec:multilane}.

\subsection{Single-Lane Setting: a Finality Gadget}
\label{subsec:singlelane}

We first present prefix consensus in the single-lane setting, where the data dissemination layer provides each process with a potentially different prefix of a hash chain.
This captures, for example, a finality gadget deployed on top of a dynamically available chain such as Ethereum's.
In this setting, proposals are chain tips and the prefix relation is the standard ancestor relation on a single hash chain.
The multi-lane machinery, deterministic interleaving, parent-relative delta encoding, and explicit skip certification introduced in Section~\ref{subsec:multilane} are unnecessary here, as the data dissemination layer already produces a single chain.

We assume a separate data dissemination layer that operates independently of the consensus protocol.
This layer is responsible for disseminating blocks and provides each process with a potentially different sequence of blocks.
Due to network delays and adversarial behavior, processes may observe different block sequences at the time of proposing: proposals may differ in length (one contains more blocks than another) or in content (distinct blocks at the same height).

A proposal can therefore be represented compactly by its tip block. We assume valid prefixes can be retrieved via a pull mechanism of the available chain, and omit calls to this mechanism (we describe in Section~\ref{sec:amplification} a push mechanism that retains total bit complexity). Formally, a value is a chain tip $b$ that implicitly encodes the sequence
\[
    \mathcal{L}(b) = (b_1,b_2,\ldots,b_h),
\]
where $h$ is the height of $b$ and each block contains the hash of its predecessor.
For two values $B$ and $B'$, we write $B \preceq B'$ when the tip of $B$ is an ancestor of the tip of $B'$, or equivalently when $\mathcal{L}(B)$ is a prefix of $\mathcal{L}(B')$.
The axioms of Section~\ref{sec:model} hold in this space: the prefixes of a tip are its ancestors, finitely many and totally ordered, with the genesis block as minimum, provided no two distinct blocks share a hash, which is where collision resistance enters.
Height realizes the $\preceq$-maximum: $\mathrm{SC}(S)$ is the highest value with support at least $2f+1$ in $S$, and $\mathsf{HCP}(S)$ the highest value that all of $S$ extends.

We define the parent induced by a single justification $J$ with endorsed value multiset $S_J$ as
\[
\textsc{ParentOf}(J)=
\begin{cases}
b^{\ast}, & \text{if } J \text{ is an M-notarization of tip } b^{\ast},\\
\mathrm{SC}(S_J), & \text{if } J \text{ is an L-notarization.}
\end{cases}
\]
Thus an M-notarization preserves the usual fast path: since all $2f+1$ votes are pairwise comparable, the next proposal must extend the highest tip among them.
An L-notarization instead requires the next proposal to extend the unique selected continuation it justifies. Two honest processes holding the same justification compute as a result the same parent.

The main protocol invokes the two-argument form $\textsc{SelectParent}(S,v)$, which selects a valid justification $J$ from view $v-1$ held in the local state $S$ and returns the pair $(\textsc{ParentOf}(J),\,J)$.
Safety does not depend on which justification is selected: the parent induced by any valid justification from view $v-1$ extends every prefix finalizable in view $v-1$.
For determinism we prescribe that $p_i$ selects an M-notarization if it holds one, and otherwise the L-notarization whose selected continuation is $\triangleleft$-minimal among those it holds.

Algorithm~\ref{alg:singlelane-hooks} shows the instantiation of these functions for the single-lane setting.

\subsubsection{Fallback Vote Selection}
\label{subsubsec:selectproposal-single}

When a process $p_j$ collects $n-f$ valid proposals in a view and the leader's proposal is unavailable, it must choose a fallback vote via \Call{SelectProposal}{$C$}, where $C$ is the set of received proposals.
Safety does not depend on the exact choice of \Call{SelectProposal}{}; it depends only on the subsequent $\mathsf{HCP}$ extraction and on \Call{SelectParent}{}.
Nevertheless, the choice affects how much prefix progress the fallback path typically achieves.

The default of Section~\ref{sec:protocol} is the $c_b=1$ instance of a \emph{ranked rule}: the backups of view $v$ are the $c_b$ indices that follow the leader in the round-robin order, and a process votes for the valid proposal of the highest-ranked backup present in $C$, falling back to its own proposal otherwise. Every process selects a proposal it already holds, so the rule introduces no constructed values.
In any view whose preferred live leader is honest and timely, all honest fallback votes coincide and its full value finalizes within the view. The self vote guarantees a votable, justified value in every view regardless, so no view is ever voided, unlike designs that restrict proposing and must void barren views by nullification. Section~\ref{sec:amplification} carries the rule over to the multi-lane setting unchanged:
the rebroadcast rules there make every fallback vote endorse a value
already amplified at its sender.

\subsubsection{Proposal Validity and Commit Rule}
\label{subsubsec:validity-single}

A proposal $B^{(v)}_j$ is \emph{valid}, in the sense checked at the vote steps of the protocol algorithm of Section~\ref{sec:protocol}, if there exists a valid justification $J^{(v-1)}$ from view $v-1$ such that
\[
    \Call{ParentOf}{J^{(v-1)}} \preceq B_j^{(v)}.
\]
If justifications are disseminated separately as shown in Algorithm~\ref{alg:hermes-psync}, a receiver does not reject a proposal merely because the corresponding justification has not yet arrived; it buffers the proposal and validates it once the relevant M-notarization or L-notarization is learned. If an implementation chooses to attach justifications explicitly, the same validity rule applies, excepting buffering.

Given any set of $4f+1$ votes from distinct processes in the same view, let $\mathcal{S}$ be the multiset of values they endorse.
A process finalizes $\mathsf{HCP}(\mathcal{S})$.
This does not remove the usual M-notarized fast path: if the $4f+1$ votes contain an M-notarization, then the M-notarized value is itself an extension of the finalized prefix extracted from that same view. The $4f+1$ votes themselves constitute a valid L-notarization for that view, and a process that finalizes from them stores and forwards them.

\begin{algorithm*}[h]
\caption{Single-lane instantiation of the value-space hooks for $p_i$}
\label{alg:singlelane-hooks}
\begin{algorithmic}[1]
\Function{ParentOf}{$J$}
  \If{$J$ is an M-notarization of tip $b^\ast$}
    \State \Return $b^\ast$
  \EndIf
  \State let $S_J$ be the multiset of values endorsed by $J$ \Comment{$J$ is an L-notarization}
  \State \Return $\mathrm{SC}(S_J)$
\EndFunction
\Function{SelectParent}{$S,v$}
  \State $J \gets$ an M-notarization from view $v-1$ in $S$ if one is held, otherwise the L-notarization from view $v-1$ whose selected continuation is $\triangleleft$-minimal among those held
  \State \Return $(\textsc{ParentOf}(J),\, J)$
\EndFunction
\Function{GetChild}{$P$}
\State let $b_i^{(v)}$ be the local available-chain tip
\If{$P \preceq b_i^{(v)}$}
\State \Return $b_i^{(v)}$
\EndIf
\State \Return $P$ \Comment{no local extension is available; re-propose the parent}
\EndFunction

\Function{SelectProposal}{$C$}
  \For{$k=1$ \textbf{to} $c_b$} \Comment{backups follow the leader in round-robin order; default $c_b=1$}
    \If{$C$ contains a valid proposal $B$ from $p_{(v+k)\bmod n}$}
      \State \Return $B$
    \EndIf
  \EndFor
\State \Return $B_i$ \Comment{the proposal $p_i$ sent in view $v$: held, valid, and broadcast in full}\EndFunction
\Function{ExtractHCP}{$V$} \Comment{$V$: $4f{+}1$ same-view votes}
  \State let $\mathcal{S}$ be the multiset of values endorsed by $V$
  \State \Return $\mathsf{HCP}(\mathcal{S})$
\EndFunction
\end{algorithmic}
\end{algorithm*}

\subsection{Multi-Lane Setting: Autobahn-Style Data Availability}
\label{subsec:multilane}

We now consider a data dissemination layer with $\hat m$ concurrent lanes, as in Autobahn.
Lane $\ell \in \{1,\ldots,\hat m\}$ produces a hash chain of blocks
\[
    \mathcal{L}_\ell = (b_{\ell,1}, b_{\ell,2}, \ldots),
\]
indexed by \emph{height}: the height of a block is its index in its lane's hash chain, assigned by the lane's producer, contiguous, and never skipped.\footnote{Height corresponds to the block number of a lane's producer in Autobahn implementations.}
We write $h_{\ell,j}$ for the hash of $b_{\ell,j}$, the block at height $j$ of lane $\ell$.
Different processes may observe different height prefixes of each lane.
The challenge is therefore not only to identify a safe prefix, but also to derive a \emph{consistent total order} across lanes.

\subsubsection{Deterministic Interleaving}
\label{subsubsec:interleaving}
We fix a global deterministic interleaving of lane positions.
For example, if the next unresolved position in lane $\ell$ is $p_\ell$, the interleaving order proceeds as
\[
    (1,p_1), (2,p_2), \ldots, (\hat m,p_{\hat m}),
    (1,p_1+1), (2,p_2+1), \ldots .
\]
Any other fixed total order would also suffice~\cite{giridharan2024autobahn}.

We call the pair $(\ell,p)$ a \emph{position}: the $p$-th slot of lane $\ell$ in the deterministic interleaving.
Positions are assigned by consensus and, unlike heights, skippable: consensus resolves every position either by binding a block of its lane to it or by explicitly certifying a skip.
The tipcut indices of Section~\ref{subsubsec:threeindex} are position-valued.
Heights and positions of a lane coincide until the lane's first skip and diverge by one per skip thereafter; Section~\ref{subsubsec:threeindex} states the exact mapping.

The key structural property is that this interleaving creates \emph{cross-lane dependencies}.
A process walking through the total order cannot finalize a later position before all earlier positions in the interleaving have already been resolved, either by committing a block there or by explicitly certifying a skip there.
This is precisely what prevents the unsafe behavior of independently finalizing each lane and only later attempting to merge the results.

\subsubsection{Parent-Relative Delta Tipcuts}
\label{subsubsec:threeindex}

To keep proposals and votes compact while still making skips explicit, we represent a multi-lane value by a \emph{parent-relative delta tipcut}.
Whenever a process validates a previous-view justification $J$, it computes
\[
    P = \Call{ParentOf}{J}
\]
and stores $P$ in a local cache under the identifier
\[
    \mathsf{pid}(P) := H(\sigma(P)),
\]
where $\sigma(P)$ is the derived sequence of $P$ (as defined below).
A proposal or vote in the next view may then refer to $P$ only through $\mathsf{pid}(P)$.
If a receiver learns the proposal before it learns the corresponding justification, it buffers the proposal until it can resolve the parent identifier, exactly as in the general implicit-justification model (and excepting buffering as in the explicit-justification one).

For process $p_i$ in view $v$, a multi-lane value has the form
\[
    T_i^{(v)}
    =
    \bigl\langle
        \mathsf{pid}_i^{(v-1)},
        (u_i^1,b_i^1,z_i^1),\ldots,(u_i^{\hat m},b_i^{\hat m},z_i^{\hat m})
      \bigr\rangle,
\]
where $\mathsf{pid}_i^{(v-1)}$ names the parent value selected from some valid justification from view $v-1$.
For each lane $\ell$:
\begin{enumerate}
    \item $u_i^\ell$ is the \emph{parent frontier} in lane $\ell$, i.e., the last position in that lane already resolved by the parent named by $\mathsf{pid}_i^{(v-1)}$;
    \item $b_i^\ell$ is the last position in lane $\ell$ for which $p_i$ proposes positive content above that parent frontier: the blocks bound there are, in order, the next $b_i^\ell - u_i^\ell$ blocks of lane $\ell$ by height;
    \item $z_i^\ell$ is the last position in lane $\ell$ that $p_i$ explicitly proposes to skip above that parent frontier.
\end{enumerate}
These indices satisfy
\[
    u_i^\ell \le b_i^\ell \le z_i^\ell.
\]
We write $\eta^\ell(P)$ for the \emph{height frontier} of the parent $P$ in lane $\ell$: the number of block entries of lane $\ell$ in $\sigma(P)$, equivalently the height of the last block of lane $\ell$ bound in $\sigma(P)$.
The position frontier $u_i^\ell$ and the height frontier $\eta^\ell(P)$ differ by exactly the number of skip entries of lane $\ell$ in $\sigma(P)$.
The positive segment of lane $\ell$ binds, in order, the blocks at heights $\eta^\ell(P)+1,\ldots,\eta^\ell(P)+(b_i^\ell - u_i^\ell)$; their hashes are carried, as usual, by a reference to the lane tip $h_{\ell,\,\eta^\ell(P)+(b_i^\ell-u_i^\ell)}$, and we suppress them in the notation.

The semantics of the suffix are \emph{blocks then skips}.
For lane $\ell$, the positions
\[
    u_i^\ell+1,\ldots,b_i^\ell
\]
contribute blocks, while the positions
\[
    b_i^\ell+1,\ldots,z_i^\ell
\]
contribute explicit skip markers.
This is the physically realizable shape of one view's contribution above a parent frontier in a hash-chained lane: available data forms a contiguous prefix, and any remaining addressed positions are explicit skips.

A receiver rejects any tipcut in which $z_i^\ell - u_i^\ell$ differs across lanes, since this would leave some interleaving positions unaddressed.

The wire format is therefore $O(\hat m)+O(1)$: one parent identifier plus three integers per lane.
The parent identifier is the semantic anchor, while the announced parent frontiers $u_i^\ell$ serve as redundant validation metadata and a fast filter for the receiver.
Section~\ref{sec:amplification} uses this terminology alongside hash-based votes to keep cubic bit complexity where $\mathcal{O}(\hat m )=\mathcal{O}(n)$.

\paragraph{Deriving a single sequence from a tipcut.}
Every parent-relative delta tipcut deterministically expands to a single sequence over block entries and skip entries.
Let $P$ be the cached parent named by $\mathsf{pid}$, and let $\sigma(P)$ denote the single sequence already resolved by that parent.
Let $\rho = z^\ell-u^\ell$ be the common reach (constant across lanes by the validity check above).
The child tipcut contributes one entry per lane at each offset $k=1,\ldots,\rho$, in the deterministic interleaving order.
For each lane, the entry is a block if the corresponding position is at most $b^\ell$, and a skip otherwise; block entries bind consecutive heights above the parent's height frontier.

\begin{algorithm*}[h]
\caption{Deriving a single sequence from a parent-relative delta tipcut}
\label{alg:derivesequence}
\begin{algorithmic}[1]
\Function{DeriveSequence}{$P,T$}
    \Comment{$T = \langle \mathsf{pid},(u^\ell,b^\ell,z^\ell)\rangle_{\ell=1}^{\hat m}$; $P$ is the cached parent named by $\mathsf{pid}$}
    \State verify that the per-lane position frontiers of $P$ equal $(u^1,\ldots,u^{\hat m})$
    \State verify that there exists $\rho \ge 0$ such that $z^\ell=u^\ell+\rho$ for every lane $\ell$
    \State $\sigma \gets \sigma(P)$
    \For{$\ell = 1$ to $\hat m$}
        \State $\eta^\ell \gets \eta^\ell(P)$ \Comment{height frontier of the parent, cached with $P$}
    \EndFor
    \For{$k = 1$ to $\rho$}
        \For{$\ell = 1$ to $\hat m$}
            \State $p \gets u^\ell + k$
            \If{$p \le b^\ell$}
                \State $\eta^\ell \gets \eta^\ell + 1$
                \State append $(\ell,p,h_{\ell,\eta^\ell})$ to $\sigma$ \Comment{block at height $\eta^\ell$ bound to position $p$}
            \Else
                \State append $(\ell,p,\bot)$ to $\sigma$ \Comment{position $p$ skipped; heights unaffected}
            \EndIf
        \EndFor
    \EndFor
    \State \Return $\sigma$
\EndFunction
\end{algorithmic}
\end{algorithm*}

We write $\sigma(T)$ for the output of \Call{DeriveSequence}{$P,T$} once the parent named by $\mathsf{pid}$ has been learned.
A process derives the sequence of a value once, upon first needing it, and caches the result; every prefix test, $\mathsf{HCP}$, $\mathrm{SC}$, and $\mathsf{pid}$ computation evaluates on cached sequences.
By definition, a later value extends an earlier one only if it preserves this entire derived sequence as a prefix.
In particular, if a position already appears in the parent as either a block or a skip, descendants do not revisit that position.
Thus a position skipped along an M-notarized chain is locked in for that chain: data arriving late for that position binds to a later position rather than being reinserted below the parent frontier.
Skipping a position of lane $\ell$ does not mean the lane stops producing, nor that any of its blocks are discarded: the block that would have occupied the skipped position binds instead to the next position resolved for that lane.
Within a lane, heights map to positions injectively and order preservingly, but not surjectively: each skip shifts every later height of the lane one position further.
Precisely, every block entry of $\sigma$ binding the block at height $j$ of lane $\ell$ to position $p$ satisfies $p = j + s$, where $s$ is the number of skip entries of lane $\ell$ at positions below $p$ in $\sigma$.

A natural question is why the parent-relative delta encoding with explicit indices is necessary, rather than a single tip index per lane as in Autobahn's native representation, with the previous-view justification attached to every proposal so that the receiver reconstructs the parent frontier and omits locally unavailable offsets. The attached encoding is simpler, one integer per lane and no parent cache, but loses three times. Without an explicit parent reference via $\mathsf{pid}$ and per-lane resolution state via $u^\ell$, prefix consensus may notarize different tipcuts across consecutive views, leaving ambiguous which interleaving history a proposal extends. Without explicit skips, two proposals differing in lane coverage at any offset derive incomparable sequences, obstructing fallback M-notarizations under heterogeneous lane speeds, and a lane missing a single position stalls entirely until the next finalization resets the base. And the inline justification costs $2f+1$ or $4f+1$ signatures per proposal where the delta proposal stays $O(\hat m)$. 
\paragraph{Example.}
We illustrate the encoding across two consecutive views: how the common reach forces an explicit skip when one lane lags, how each tipcut expands via \Call{DeriveSequence}{} into the interleaved sequence, and how a skipped position, once below the parent frontier, is locked in while the lane's late block binds to a later position.

Consider $\hat m = 2$ lanes and a parent $P$ (named by $\mathsf{pid}$) whose derived sequence contains no skips, so that position and height frontiers coincide: $u^1 = \eta^1 = 2$ and $u^2 = \eta^2 = 0$.

\medskip\noindent\emph{View $v$.}
process $p_i$ holds lane~1 through height~4 and lane~2 through height~1.
The common reach forces $\rho = 2$: lane~1 fills both offsets with blocks, lane~2 fills one.
The tipcut is
\[
    T_i = \bigl\langle \mathsf{pid},\;
    (u^1{=}2,\, b^1{=}4,\, z^1{=}4),\;
    (u^2{=}0,\, b^2{=}1,\, z^2{=}2)
    \bigr\rangle.
\]
Lane~1 binds the blocks at heights~3 and~4 to positions~3 and~4.
Lane~2 binds the block at height~1 to position~1 and skips position~2.
The suffix appended to $\sigma(P)$ is
\[
    (1,3,h_{1,3}),\; (2,1,h_{2,1}),\;
    (1,4,h_{1,4}),\; (2,2,\bot).
\]

\medskip\noindent\emph{View $v{+}1$.}
Suppose an M-notarization forms on $T_i$ in view~$v$, so that $T_i$ is the parent returned by \Call{ParentOf}{} for the next view, with position frontiers $u^1=4,\,u^2=2$ and height frontiers $\eta^1=4,\,\eta^2=1$: the frontiers of lane~2 have diverged by its one skip.
process $p_i$ now holds lane~1 through height~6 and lane~2 through height~2: the block at height~2, too late for position~2, has arrived.
With $\rho=2$, the tipcut is
\[
    T_i' = \bigl\langle \mathsf{pid}',\;
    (u^1{=}4,\, b^1{=}6,\, z^1{=}6),\;
    (u^2{=}2,\, b^2{=}3,\, z^2{=}4)
    \bigr\rangle.
\]
Lane~1 binds heights~5 and~6 to positions~5 and~6.
Lane~2 binds the block at height~2 to position~3, one position after its height by the invariant above ($3 = 2 + 1$ skip), and skips position~4.
The appended suffix is
\[
    (1,5,h_{1,5}),\; (2,3,h_{2,2}),\;
    (1,6,h_{1,6}),\; (2,4,\bot).
\]
The skip at position~2 is locked in below the parent frontier; no block of lane~2 is discarded, and the lane continues at later positions with its heights intact.

\subsubsection{Prefix Relation and Selected Continuation for Tipcuts}
\label{subsubsec:prefix-multi}

For two tipcuts $T$ and $T'$, we write $T \preceq T'$ if and only if $\sigma(T)$ is a prefix of $\sigma(T')$.

Given $4f+1$ votes from the same view carrying tipcuts $T_1,\ldots,T_{4f+1}$, let $\Sigma$ be the multiset of their derived sequences.
The definitions of Section~\ref{sec:model} apply verbatim: $\mathsf{HCP}(\Sigma)$ is the longest sequence that every sequence in $\Sigma$ extends, and $\mathrm{SC}(\Sigma)$ is the longest sequence with support at least $2f+1$ in $\Sigma$, both unique, length realizing the $\preceq$-maximum.

This reduces the multi-lane instantiation to the single-lane one: the value-space hooks are those of Algorithm~\ref{alg:singlelane-hooks} applied to resolved sequences. \textsc{ParentOf} returns the $\preceq$-maximal resolved sequence endorsed by an M-notarization, and $\mathrm{SC}(\Sigma_J)$ over the resolved multiset of an L-notarization, together with its per-lane position and height frontiers; \textsc{SelectParent} selects the justification exactly as in Section~\ref{subsec:singlelane}, applying a justification once every tipcut it endorses resolves, and stores the resulting parent in the cache under $\mathsf{pid}$ as prescribed in Section~\ref{subsubsec:threeindex}; \textsc{SelectProposal} is the ranked rule over the valid proposals received; and the commit rule finalizes $\mathsf{HCP}$ over the resolved sequences of any $4f+1$ same-view votes. The cache is initialized with the genesis record, the empty sequence with all frontiers zero, which justifies view $1$.

Only \textsc{GetChild} requires a new construction: to extend a parent $P$, the proposer sets $b^\ell = u_P^\ell + n^\ell$ for each lane $\ell$, where $n^\ell$ counts the consecutive blocks of lane $\ell$ it holds at heights above $\eta^\ell(P)$, and pads every lane to the common reach $\rho = \max_\ell n^\ell$ by $z^\ell = u_P^\ell + \rho$; a view with $\rho = 0$ re-proposes the parent frontier.

\subsubsection{Proposal Validity for Tipcuts}
\label{subsubsec:validity-multi}

A multi-lane proposal
\[
    T_j^{(v)} = \bigl\langle \mathsf{pid},(u_j^\ell,b_j^\ell,z_j^\ell)\bigr\rangle_{\ell=1}^{\hat m}
\]
is valid if there exists a valid justification $J^{(v-1)}$ from view $v-1$ such that, letting
\[
    P = \Call{ParentOf}{J^{(v-1)}},
\]
all of the following hold:
\begin{enumerate}
    \item the proposal carries $\mathsf{pid}(P)$ as its parent identifier;
    \item the advertised position frontiers $u_j^\ell$ match those of $P$ for every lane $\ell$;
    \item there exists a common reach $\rho \ge 0$ such that $z_j^\ell=u_j^\ell+\rho$ for every lane $\ell$;
    \item the positive segment of each lane $\ell$ binds, in order, the blocks at heights $\eta^\ell(P)+1,\ldots,\eta^\ell(P)+(b_j^\ell-u_j^\ell)$ of lane $\ell$'s hash chain; and
    \item under the prefix order on derived sequences,
    \[
        \sigma(P) \text{ is a prefix of } \sigma(T_j^{(v)}).
    \]
\end{enumerate}
Conditions 2, 3, and 5 are enforced by \Call{DeriveSequence}{} itself: it checks the frontiers and the common reach, and it constructs $\sigma(T_j^{(v)})$ by appending to $\sigma(P)$, so condition 5 holds by construction whenever derivation succeeds.
A tipcut is therefore valid for view $v$ if and only if its resolution succeeds and its parent identifier is cached under a justification from view $v-1$; the cache records the views of the inducing justifications precisely to enforce this, since a Byzantine proposer may name a genuinely cached parent from an older view.
A proposal whose parent identifier is unknown is \emph{pending}: the receiver buffers it until the parent is cached.
A proposal that fails any check of \Call{DeriveSequence}{} or condition 4 is \emph{invalid} and discarded.

\subsubsection{Amplification of Proposals}
\label{sec:amplification}

We say a vote is \emph{resolved} at a process once that process holds the endorsed value, in full or by reconstruction. In the good case of \consensus, the leader is honest after GST, votes carry hashes and every value the protocol acts on is backed by a quorum of matching votes, hence by at least $f+1$ honest processes holding the value in full, so a missing value can always be pulled and optimistic waiting costs latency at most. Prefix consensus severs this link between agreement support and availability: the heaviest common prefix is extracted from $4f+1$ votes that need not match. A value may then be used after it is supported by only one vote, potentially by a Byzantine process that may never make the value available.

Since the Byzantine may also make the value only available to a subset of processes, proposers may have to serve values they used for justifications, incurring $\mathcal{O}(n)$ serves per process in the worst case. However, carrying values in full wherever they are referenced also inflates justification forwarding to $O(n^3\hat m)$ bits, also quartic for $\hat m = \Theta(n)$. 

We propose a solution to these challenges via having each process amplify all received proposals in the fallback case via erasure codes, maintaining cubic bit complexity.

\paragraph{Coded values and identifiers.} Fix a deterministic $(t,n)$ erasure code with reconstruction threshold $f+1 \leq t\le 3f+1$ and a binding vector commitment over its $n$ symbols (e.g., a Merkle tree; we write $O(\lambda)$ for an opening, absorbing the logarithmic factor). For a value $T$, let $\mathrm{frag}_k(T)$ be the symbol at index $k$ of the encoding of the canonical serialization of $T$, and let $\mathrm{id}(T)$ be the commitment root over the $n$ symbols. Since the encoding is deterministic, $\mathrm{id}$ is a collision-resistant hash of $T$, and we instantiate every hash of a value in the protocol, including the matching of votes to proposals, as $\mathrm{id}(T)$, so that fragment openings verify against the identifiers votes already carry. Equivocation produces distinct identifiers and is handled by the unchanged quorum logic.

\paragraph{Votes and justifications.} Proposals travel in full in \textsc{Propose}, in the delta tipcut format of Section~\ref{subsubsec:threeindex}. A vote for $B_j$ carries $(\mathrm{id}(B_j))$. Forwarded M-notarizations and L-notarizations are sets of (view, sender, identifier, signature) tuples of size $O(n\lambda)$.

\paragraph{Amplification.}
Fragments are sender indexed: process $p_k$ only ever contributes its own fragment $\mathrm{frag}_k(T)$, and a receiver accepts $(\textsc{amplify}, \mathrm{id}(T), \mathrm{frag}_k(T))$ only from $p_k$, only with a valid opening against $\mathrm{id}(T)$, and at most once per identifier and sender. A full copy of $T$ received from $p_k$ counts as $p_k$'s fragment. A value is \emph{resolved} at a process once the process holds it in full or holds accepted fragments for it from $t$ distinct senders, and \emph{amplified} once the process holds accepted fragments for it from $t+f$ distinct senders or has itself broadcast it in full. Amplification is the transferable notion: $t+f$ accepted senders certify $t$ honest broadcasts, which reach every process; resolution via $t$ fragments certifies as little as one honest broadcast and does not transfer.
Four rules generate amplification traffic:
\begin{enumerate}
\item upon first receiving the view leader's proposal in full, a process
rebroadcasts it in full, at most once per view;
\item upon first receiving a backup proposal in full at or after its
$2\Delta$ timeout, or at the timeout for backup proposals already held, a
process rebroadcasts it in full, at most once per view and backup proposal;
\item every \textsc{Vote} message of $p_i$ additionally carries
$\mathsf{frag}_i(T)$ for every view value $T$ resolved at $p_i$ at sending
time, each identifier at most once ever;
\item at its $2\Delta$ timeout, if it has not left the view, a process
sends each other process a single \textsc{Amplify} message batching
$\mathsf{frag}_i(T)$ for every view value $T$ resolved locally whose
fragment it has not yet sent; and, at any time, $p_i$ broadcasts
$\mathsf{frag}_i(T)$, at most once per identifier ever, upon holding a
value $T$ that is resolved but not amplified and whose identifier is
endorsed by a received vote or notarization.
\end{enumerate}
In the good case the leader is timely: every honest process receives the
leader's proposal directly and rebroadcasts it (rule~1), so the value is
amplified at each process by its own rebroadcast and at every process
within one delivery via $4f{+}1\geq t{+}f$ full-copy senders. The M-notarization is applied before the timeouts, so rule~2 and the
timeout batch of rule~4, which fire only at a timeout reached inside the
view, stay silent, up to the race between an aggressive timeout and the
M-notarization's third delivery; the endorsement relay of rule~4 fires
only for a value that a received vote or notarization endorses before it
amplifies, which does not occur without Byzantine votes. Amplification adds one $O(n^2c)$ rebroadcast and the
piggybacked fragments of rule~3, $O(n(\hat m/t+\lambda))=O(\hat
m+n\lambda)$ per vote, within the \textsc{Propose} and vote budgets.

\paragraph{Fallback selection.}
In the multi-lane instantiation, \textsc{SelectProposal} may select a backup proposal only once its value is amplified at the voter. At the $2\Delta$ timeout, rule~2 runs before selection, so the voter's full-value rebroadcast amplifies every held backup proposal. Its own proposal is already amplified by its own \textsc{Propose} broadcast, ensuring an amplified fallback is always selectable.

\paragraph{Justification.} A process acts on a vote, proposal, or notarization only once the values it depends on are amplified, as the delivery gating below makes precise. This is the availability analogue of the validity rules of Section~\ref{subsubsec:validity-multi}: anything an honest process justifies with, every process can reconstruct. Messages that cannot be delivered are instead buffered until the corresponding values have been amplified.

\paragraph{Jump votes.} A jump vote carries $\mathrm{id}(\SC(S_L))$. Any process holding $L$ computes the derived sequence of $\SC(S_L)$ deterministically from the resolved values of $L$, encodes it, and checks the identifier; the vote is treated as amplified once every value endorsed in $L$ is amplified.

\paragraph{Delivery gating and storage.} A process tallies a vote once its endorsed value is amplified, applies a proposal once every value endorsed by the votes of a justification
validating it is amplified, and applies a notarization, at any guard of
Algorithm~\ref{alg:hermes-psync}, once every value its own votes endorse is
amplified; messages that cannot yet be processed are buffered. Buffered state per view is at most one proposal and one vote per sender, $O(n)$ values of $O(\hat m)$ bits each, garbage collected when a later view finalizes.\section{Complexity \& Latency}
\label{sec:complexity}

Table~\ref{tab:complexity} summarizes the content.

\begin{table*}[t]
\centering
\caption{Per-view communication and post-GST finalization latency of \consensus, measured from view entry; $\delta \le \Delta$ is the actual network delay. Multi-lane values are tipcuts of size $O(\hat m)$; for $\hat m = \Theta(n)$ the amplified protocol yields $O(\lambda n^3)$.}
\label{tab:complexity}
\begin{tabular}{@{}lcccc@{}}
\toprule
& \multicolumn{2}{c}{Communication} & \multicolumn{2}{c}{Finalization latency} \\
\cmidrule(lr){2-3}\cmidrule(l){4-5}
Instantiation & Bits & Messages & Timely leader & Fallback/Prefix\\
\midrule
Generic, values of size $c$ & $O(n^2 c + \lambda n^3)$ & $O(n^2)$ & $2\delta$ & $2\Delta+\delta$ \\
Single lane ($c = O(\lambda)$) & $O(\lambda n^3)$ & $O(n^2)$ & $2\delta$ & $2\Delta+\delta$ \\
Multi lane, full-value forwarding & $O(n^3 \hat m)$ & $O(n^2)$ & $2\delta$ & $2\Delta+\delta$ \\
Multi lane, amplified & $O(n^2 \hat m + \lambda n^3)$ & $O(n^2)$\,\textsuperscript{\textdagger} & $2\delta$ & $2\Delta+\delta$ \\
\bottomrule
\end{tabular}\\
\textdagger\,$O(n^3)$ under faults, by rule~4 of Section~\ref{sec:amplification}.
\end{table*}
\paragraph{Communication.}
Per view, \textsc{Propose} is $O(n^2)$ messages of $O(c)$ payload; votes
are $O(n^2)$ messages, of $O(\lambda)$ payload in the generic and
single-lane settings; and justification forwarding dominates, an
M-notarization or L-notarization being $O(n)$ (view, sender, identifier,
signature) tuples of $O(n\lambda)$ total, forwarded at most once per kind
per process per view, for $O(\lambda n^3)$. The generic total is
$O(n^2c+\lambda n^3)$ bits and $O(n^2)$ messages. The single-lane gadget is
the case $c=O(\lambda)$.

In the multi-lane setting $c=O(\hat m)$, and the choice is what
notarizations carry: full tipcuts inflate forwarding to $O(n^3\hat m)$,
quartic for $\hat m=\Theta(n)$, while the amplification of Section~\ref{sec:amplification}
keeps identifiers everywhere. Its traffic decomposes by rule. The
rebroadcasts of rules~1 and~2 cost $O((c_b{+}1)\,n^2\hat m)$, within the
\textsc{Propose} budget, and only rule~1 fires in the good case. The
piggybacked fragments of rule~3 add $O(n(\hat m/t+\lambda))=
O(\hat m+n\lambda)$ to each vote for $t=\Theta(n)$, for
$O(n^2\hat m+\lambda n^3)$ over all votes. Rule~4's timeout batch is one message per process per view, $O(n^2)$
messages of $O(\hat m+n\lambda)$ bits, sent only when a timeout fires
inside the view; its endorsement clause fires at most once per identifier
per process, over $O(n)$ view identifiers, for $O(n^3)$ messages of
$O(\hat m/t+\lambda)$ bits under faults and none in the good case. Total: $O(n^2\hat m+\lambda n^3)$ bits; $O(n^2)$
messages in the good case, $O(n^3)$ under faults.

\paragraph{Latency.} After GST a timely leader finalizes in two message delays, \textsc{Propose} then \textsc{Vote}, for $2\delta$; the votes that finalize also justify the next view, so views pipeline with no view-change round. Under a faulty or slow leader, processes vote at the $2\Delta$ timeout and the heaviest common prefix of those votes finalizes one delivery later, for $2\Delta + \delta$: the proposals voted on were delivered before the timeout by construction, and their amplification evidence, the backup rebroadcasts of rule~2 and the fragments riding the votes themselves (rule~3), arrives with the votes:
after GST every honestly proposed value is delivered to every honest
process before its timeout, so votes for honestly proposed values, own or
backup, tally on arrival. In the single-lane gadget, missing ancestry resolves through the chain's own request/response sync, votes buffering until their tips validate and tips cited only by Byzantine processes being excluded in favor of the $4f+1$ honest resolvable votes; a required pull adds its round trip to the tabulated latency, arises only under selective Byzantine delivery, and costs the adversary neither safety nor liveness.

In the multi-lane setting the residual corner is a straggler or catch-up
vote on a Byzantine-crafted tip that no honest process ever received as a
proposal, or a jump vote whose L-notarization endorses such a value: rules~1
and~2 never fire for it, so it amplifies at the remaining processes only
through the endorsement relay of rule~4, one delivery after the endorsing
vote or forwarded notarization arrives
(Corollary~\ref{cor:amplified}). Fallback votes never lag, even on a backup
proposal delivered selectively by a Byzantine sender: the holders' rule~2
rebroadcasts travel with their votes, and every receiver of a full copy is
amplified by its own rebroadcast.
\section{Proofs}
\label{sec:proofs}

Throughout, votes, notarizations, and finalizations refer to the protocol of Algorithm~\ref{alg:hermes-psync} over an arbitrary value space satisfying the axioms of Section~\ref{sec:model}; by the reduction of Section~\ref{subsubsec:prefix-multi}, every statement applies verbatim to the single-lane instantiation over chain tips and to the multi-lane instantiation over resolved derived sequences.

A view-$v$ vote is \emph{counted} if its endorsed value is a valid view-$v$ value: one that extends $\Call{ParentOf}{J}$ for some valid view-$(v-1)$ justification $J$. Honest votes are counted by construction: leader-path and fallback votes endorse valid proposals, and carry-over votes endorse either the $\preceq$-maximum of a counted M-notarization, itself a counted vote's value, or the heaviest common prefix of a counted view-$v$ L-notarization, which extends the selected continuation of the justification behind any proposal voted within it. A vote whose validity cannot yet be verified is buffered, exactly as proposals are; notarizations and prefix extraction range over counted votes only.\footnote{Counting is not optional: since $\mathsf{HCP}$ requires full support, a single Byzantine vote endorsing the genesis value inside a $4f{+}1$ set pins its extraction to genesis, and injecting $f$ such votes per view voids all progress forever. The selected continuation is immune, its $2f{+}1$ threshold tolerating $f$ junk votes, which is why parent selection never needed the restriction; extraction does.}

We say a value $x$ is \emph{finalizable in view $v$} if some set $V$ of $4f+1$ counted view-$v$ votes from distinct processes satisfies $x \preceq \mathsf{HCP}(\mathcal{S}_V)$, where $\mathcal{S}_V$ is the multiset of values $V$ endorses; such a set is a valid L-notarization and we call it a \emph{witness} for $x$. A value is \emph{finalized} when some honest process appends it (line~\ref{line:append}); every finalized value is finalizable in the view of its L-notarization. Recall from Section~\ref{sec:model} that for any multiset $\mathcal{S}$ of at most $4f+1$ values, the set $\{x : \mathrm{supp}_{\mathcal{S}}(x) \ge 2f+1\}$ is totally ordered by $\preceq$ with maximum $\mathrm{SC}(\mathcal{S})$, and that every honest process casts at most one vote per view, every vote send being guarded by $voted_v$.

\subsection{Safety}
\label{subsec:safety-prefix}
\begin{lemma}[Vote support of finalizable prefixes]
\label{lem:vote-support}
Let $x$ be finalizable in view $v$ and let $W$ be any set of $m$ counted view-$v$ votes from distinct processes. Then $W$ contains at least $m - 2f$ votes cast by honest processes whose endorsed values extend $x$. In particular, an L-notarization ($m = 4f+1$) contains at least $2f+1$ such votes, and an M-notarization ($m = 2f+1$) contains at least one.
\end{lemma}
\begin{proof}
Let $V$ witness the finalizability of $x$; every value endorsed in $V$ extends $\mathsf{HCP}(\mathcal{S}_V) \succeq x$. The sets $V$ and $W$ name $|V| + |W| = 4f+1+m$ processes among $n = 5f+1$, so at least $m - f$ distinct processes appear in both, of whom at least $m - 2f$ are honest. An honest process casts at most one view-$v$ vote, so for each such process its vote in $W$ \emph{is} its vote in $V$, and its value extends $x$.
\end{proof}

\begin{lemma}[Justified proposals extend finalizable prefixes]
\label{lem:parent-safety}
Let $x$ be finalizable in view $v$ and let $J$ be any valid view-$v$ justification. Then $x \preceq \Call{ParentOf}{J}$. Consequently every valid proposal of view $v+1$ extends $x$.
\end{lemma}
\begin{proof}
The second claim follows from the first by the validity rule of Section~\ref{subsubsec:validity-single} and transitivity of $\preceq$: a valid view-$(v+1)$ proposal extends $\Call{ParentOf}{J}$ for some valid view-$v$ justification $J$. For the first claim we distinguish the two kinds of justification.

\emph{Case L-notarization.} $J$ consists of $4f+1$ counted view-$v$ votes with endorsed multiset $\mathcal{S}_J$, and $\Call{ParentOf}{J} = \mathrm{SC}(\mathcal{S}_J)$. By Lemma~\ref{lem:vote-support}, $\mathrm{supp}_{\mathcal{S}_J}(x) \ge 2f+1$, so $x$ lies in the $\preceq$-totally-ordered set $\{y : \mathrm{supp}_{\mathcal{S}_J}(y) \ge 2f+1\}$, whose $\preceq$-maximum is $\mathrm{SC}(\mathcal{S}_J)$.

\emph{Case M-notarization.} $J$ consists of $2f+1$ pairwise comparable counted view-$v$ votes, and $\Call{ParentOf}{J} = b^\ast$ is their $\preceq$-maximum. By Lemma~\ref{lem:vote-support}, some vote in $J$ endorses a value $B$ with $x \preceq B$; and $B \preceq b^\ast$ by maximality. Hence $x \preceq b^\ast$ by transitivity.
\end{proof}

\begin{lemma}[Finalizable prefixes of one view are comparable]
\label{lem:same-view}
Any two values finalizable in the same view are comparable.
\end{lemma}
\begin{proof}
Let $x_1, x_2$ be finalizable in view $v$ and let $V_2$ witness $x_2$. By Lemma~\ref{lem:vote-support}, some vote in $V_2$ endorses a value $B$ extending $x_1$; and $x_2 \preceq \mathsf{HCP}(\mathcal{S}_{V_2}) \preceq B$ since every value of $V_2$ extends the extraction. Both $x_1$ and $x_2$ are prefixes of the single value $B$, hence comparable by the model axioms.
\end{proof}

\begin{lemma}[Domination persists across views]
\label{lem:chain}
If $x$ is finalizable in view $v$, then for every view $w > v$, every counted view-$w$ vote endorses a value extending $x$. Consequently $x \preceq \mathsf{HCP}(\mathcal{S}_W)$ for every valid L-notarization $W$ of every view $w > v$, and $x \preceq \Call{ParentOf}{J}$ for every valid justification $J$ of every view $u \ge v$.
\end{lemma}
\begin{proof}
By induction on $w$. A counted view-$w$ vote endorses a valid view-$w$ value, which extends $\Call{ParentOf}{J}$ for some valid view-$(w-1)$ justification $J$; so it suffices that $x \preceq \Call{ParentOf}{J}$ for every valid view-$(w-1)$ justification. For $w-1 = v$ this is Lemma~\ref{lem:parent-safety}. For $w-1 > v$, the induction hypothesis gives that all $4f+1$ counted votes of any view-$(w-1)$ L-notarization endorse values extending $x$, so $x$ has full support in its endorsed multiset and lies below its $\mathsf{HCP}$, hence below its $\mathrm{SC}$; and the $\preceq$-maximum of any view-$(w-1)$ M-notarization is itself a counted vote's value, extending $x$. Either parent therefore extends $x$, completing the induction; the two consequences restate its last step.
\end{proof}

\begin{theorem}[Safety]
\label{thm:safety}
No two honest processes ever finalize incomparable values.
\end{theorem}
\begin{proof}
We show that any two finalizable values, in any views, are comparable; finalized values are finalizable, giving the claim. Let $x$ be finalizable in view $v$ and $y$ in view $w \ge v$ with witness $W$. If $w = v$, this is Lemma~\ref{lem:same-view}. If $w > v$, then $x \preceq \mathsf{HCP}(\mathcal{S}_W)$ by Lemma~\ref{lem:chain} and $y \preceq \mathsf{HCP}(\mathcal{S}_W)$ by definition of witness, so both are prefixes of one value and comparable.
\end{proof}

\subsection{Liveness}
\label{subsec:liveness}
In the amplified multi-lane instantiation, a vote tallies once its endorsed
value is amplified at the tallying process. Lemma~\ref{lem:availability}
bounds when this happens: a value broadcast in full by some honest process
is amplified everywhere no later than the votes endorsing it arrive, and
any endorsed value is amplified everywhere within one delivery of its
vote's delivery.Lemmas~\ref{lem:sync} to~\ref{lem:honest-leader} are stated for the
protocol in which a delivered message is immediately usable. This covers
the generic protocol and the single-lane instantiation up to the pull round
trips of Section~\ref{sec:complexity}; the amplified multi-lane
instantiation is treated in Section~\ref{subsec:liveness-amplified}, whose
Corollary~\ref{cor:amplified} transfers the bounds.

\begin{lemma}[View synchronization]\label{lem:sync} 
If the first honest process reaches a view at least $v$ at time $t$ after GST, every honest process is in a view at least $v$ by $t+\delta$.
\end{lemma}
\begin{proof}
The entrant entered upon verifying a justification for some view $v' \ge v{-}1$ (lines~\ref{line:mnot-jump-guard} and~\ref{line:lnot-jump-guard}), which every honest process verifies by $t + \delta$; the jump guard then places every process in a view at least $v'{+}1 \ge v$.
\end{proof}
\begin{lemma}[View completion]\label{lem:completion}
If the first honest process reaches a view at least $v$ at time $t$ after GST, then there is a view $w \ge v$ such that, by $t+2\Delta+2\delta$, every honest process holds a valid view-$w$ L-notarization and its finalized head extends the heaviest common prefix of one.
\end{lemma}

\begin{proof}
By Lemma~\ref{lem:sync}, every honest process is in a view at least $v$ by $t+\delta$. Views only increase, so an honest process whose view ever exceeds $v$ without equaling $v$ crossed $v$ by $t+\delta$, upon verifying a justification for some view $w' \ge v$. Verifying the votes of a justification for view $w'$ requires holding a valid justification for view $w'{-}1$, hence, inductively, a valid justification for every view in $[v, w'{-}1]$, in particular for view $v$ itself.

Case 1: some honest process verifies a valid L-notarization for some view $w \ge v$ by $t+2\Delta+\delta$. It forwards it upon verification, and it forwarded each element of the justification chain validating its votes upon first verifying that element; every honest process therefore verifies the L-notarization by $t+2\Delta+2\delta$ and applies the finalization rule to it, after which its head extends the extracted prefix whether or not the append fires. Set $w$ to that view.

Case 2: otherwise. We claim every honest process casts a counted view-$v$ vote by $t+2\Delta+\delta$.

Suppose first some honest process $q$ crosses $v$ without entering it. By the opening paragraph $q$ verifies a valid view-$v$ justification by $t+\delta$, which in Case 2 is an M-notarization. Upon its first such verification, the straggler guard (lines~\ref{line:mnot-guard} to~\ref{line:send-votemnot}) makes $q$ vote in view $v$ for its tip if $q$ has not voted in $v$, and $q$ forwards it with its validating chain. Every honest process verifies that M-notarization by $t+2\delta$ and, if it has not voted in view $v$, votes for its tip then: all honest view-$v$ votes are cast by $t+2\delta$.

Otherwise every honest process enters view $v$, by $t+\delta$, and proposes at entry. All $4f{+}1$ honest proposals, and the justifications validating them, are delivered by $t+2\delta$, before any honest timeout, since each timeout fires $2\Delta \ge 2\delta$ after an entry no earlier than $t$. An honest process still in view $v$ at its timeout, which fires by $t+2\Delta+\delta$, holds the leader's valid proposal or at least $n{-}f$ valid proposals, and votes. An honest process that left view $v$ before its timeout verified a justification for a view at least $v$ by then, hence, by the opening paragraph and Case 2, a view-$v$ M-notarization no later, and voted in view $v$ at that verification via the straggler guard.

In both subcases every honest vote is counted: leader-path and fallback votes endorse valid proposals, and straggler votes endorse the tip of a valid M-notarization. The $4f{+}1$ votes, validated at receivers by the justification chains forwarded no later than the votes themselves, are delivered by $t+2\Delta+2\delta$; they form a valid view-$v$ L-notarization at every honest process, which imposes no agreement among its votes, and every honest process applies the finalization rule to it. Set $w=v$.
\end{proof}
\begin{lemma}[Liveness under an honest leader]\label{lem:honestleader}\label{lem:honest-leader} 
If the leader $p_{L_v}$ of view $v$ is honest and the first honest process reaches a view at least $v$ at time $t$ after GST, then every honest process's finalized head extends the leader's proposal $B_{L_v}$ by $t+3\delta$.
\end{lemma}

\begin{proof}
By Lemma~\ref{lem:sync}, every honest process is in a view at least $v$ by $t+\delta$; the leader enters $v$ and proposes by $t+\delta$, so $B_{L_v}$ is delivered to every honest process by $t+2\delta$. Every honest timeout fires $2\Delta \ge 2\delta$ after an entry no earlier than $t$, so $B_{L_v}$ arrives before any honest view-$v$ timeout and no honest fallback vote fires.

Every honest view-$v$ vote endorses a value extending $B_{L_v}$, by strong induction on the order in which honest view-$v$ votes are cast. A leader-path vote endorses $B_{L_v}$. A straggler or catch-up vote endorses the tip of a valid view-$v$ M-notarization, which contains at least $f{+}1$ honest view-$v$ votes cast earlier, each endorsing an extension of $B_{L_v}$ by the induction hypothesis; the tip, the $\preceq$-maximum of pairwise comparable values, extends one of them, hence extends $B_{L_v}$. A jump vote is cast upon verifying a view-$v$ L-notarization, which contains at least $3f{+}1$ honest votes, all extending $B_{L_v}$ by the induction hypothesis, so $B_{L_v}$ has support at least $2f{+}1$ in its endorsed multiset and the vote's value $\SC \succeq B_{L_v}$; such a vote is moreover counted, since $\SC$ extends the leader's parent.

Every honest process votes in view $v$ by $t+2\delta$: a process in view $v$ at $t+2\delta$ holds $B_{L_v}$ and votes leader-path if it has not voted; a process beyond view $v$ crossed it by $t+\delta$ (Lemma~\ref{lem:sync}), holds a valid view-$v$ justification at the crossing by the chain argument of Lemma~\ref{lem:completion}, and voted in view $v$ upon first verifying it, via the straggler guard for an M-notarization or the jump guard for an L-notarization.

The $4f{+}1$ honest view-$v$ votes, and
their validating justifications are delivered by $t+3\delta$. They constitute a valid view-$v$ L-notarization in which every endorsed value extends $B_{L_v}$, so its heaviest common prefix extends $B_{L_v}$; every honest process applies the finalization rule to it by $t+3\delta$, after which its head extends $B_{L_v}$ whether or not the append fires.
\end{proof}

Lemmas~\ref{lem:completion} and~\ref{lem:honest-leader} hold verbatim for
the amplified multi-lane instantiation whenever every vote of the
finalizing L-notarization endorses a value that some honest process
broadcast in full, by Lemma~\ref{lem:availability}(i); otherwise their
bounds hold with one additional $\delta$, by
Lemma~\ref{lem:availability}(ii).

\subsection{Liveness of the Amplified Instantiation}
\label{subsec:liveness-amplified}

In the amplified multi-lane instantiation, a vote tallies once its endorsed
value is amplified at the tallying process. Two facts govern the timing.
First, resolution is never the bottleneck: the gating rule admits an honest
vote only once its endorsed value is amplified at the voter, which
certifies $t$ honest fragment broadcasts already in flight, so every honest
process can reconstruct the value by the vote's delivery. Second, the tally
requires amplification rather than resolution, because only amplification
transfers: a process that finalizes from an L-notarization forwards it, and
its receivers must themselves certify availability.

\begin{lemma}[Amplification transfer, multi-lane]
\label{lem:transfer}
An identifier amplified at an honest process, and endorsed by a message
delivered to all honest processes, is amplified at every honest process
within two deliveries.
\end{lemma}
\begin{proof}
Of the $t{+}f$ senders accepted at amplification, at most $f$ are
Byzantine, so at least $t$ fragments or full copies were broadcast by
honest processes and are accepted everywhere within one delivery, resolving
the value; each holder then relays its own fragment (rule~4), and the
$4f{+}1$ honest fragments are accepted everywhere one delivery later. In
the full-value case the broadcast itself resolves everyone and the same
relay closes it.
\end{proof}

\begin{lemma}[Value availability, multi-lane]
\label{lem:availability}
After GST, let an honest vote endorse value $T$.
(i) If every honest process resolves $T$ no later than casting its vote of
$T$'s view or reaching a $2\Delta$ timeout inside that view, then $T$ is
amplified at every honest process within one delivery of the latest such
vote or timeout.
(ii) If $T$ is the view leader's proposal or a backup proposal and some
honest process holds $T$ in full, then every honest process is amplified
for $T$ no later than the later of its first receipt of a full copy and its
own view timeout, and in particular no later than the delivery of any
honest vote endorsing $T$.
(iii) In any case, $T$ is amplified at every honest process within one
delivery of the vote's delivery.
\end{lemma}
\begin{proof}
(i) Each honest $p_i$ sends $\mathsf{frag}_i(T)$ to all no later than that
point: with its vote if $T$ is resolved by then (rule~3), and in its
timeout batch otherwise (rule~4). The $4f{+}1 \geq t{+}f$ honest fragments
are accepted everywhere one delivery later.

(ii) An honest holder rebroadcasts $T$ in full, on first receipt for the
leader's proposal (rule~1) and no later than its view timeout for a backup
proposal (rule~2), and its vote endorsing $T$ is sent no earlier. Full
copies therefore reach every honest process no later than the votes. A
receiver of a full copy rebroadcasts by the same rules, immediately for the
leader's proposal and no later than its own timeout for a backup one, and
is amplified by its own broadcast. Votes endorsing a backup proposal are
sent at their senders' timeouts and honest timeouts differ by at most
$\delta$ (Lemma~\ref{lem:sync}), so such votes arrive no earlier than the
receiver's timeout. View timers are wall clock and per view, so rules~1
and~2 apply also at processes that have advanced past the view.

(iii) A leader-path or fallback vote endorses a value broadcast in full by
the voter itself: its own proposal at \textsc{Propose}, or the leader's or
a backup leader's proposal rebroadcast by rules~1 and~2. A straggler,
catch-up, or jump vote is cast upon applying a notarization, which the
gating rule admits only once every value its votes endorse is amplified at
the applier; the endorsed value of the cast vote is one of these, the
$\preceq$-maximum of the M-notarization for a straggler or catch-up vote,
and $\mathrm{SC}(\mathcal{S}_L)$, computable from the endorsed values of
$L$, for a jump vote. In all cases $T$ is amplified at the voter no later
than the vote's sending, and the vote is broadcast to all honest processes,
so Lemma~\ref{lem:transfer} amplifies $T$ everywhere within two deliveries
of the sending, hence within one delivery of the vote's delivery.
\end{proof}

\begin{corollary}[Amplified multi-lane instantiation]
\label{cor:amplified}
Lemmas~\ref{lem:completion} and~\ref{lem:honest-leader} hold verbatim for
the amplified multi-lane instantiation whenever every vote of the
finalizing L-notarization is covered by Lemma~\ref{lem:availability}(i)
or~(ii); otherwise their bounds hold with at most one additional $\delta$,
by Lemma~\ref{lem:availability}(iii).

After GST every leader-path and fallback vote is covered: it endorses the
leader's, a backup's, or the voter's own proposal; honestly proposed values
are delivered everywhere before the timeouts, and selectively delivered
ones fall under part~(ii). A vote escapes coverage only at the straggler,
catch-up, or jump guards, on a value that no honest process ever received
as a proposal, which requires Byzantine-crafted values. In particular the
bounds hold verbatim in every execution in which the faulty processes
merely crash.
\end{corollary}

\section{Related Work}
\label{sec:related}

\paragraph{Leader-based BFT and two-round finality.}
Classical leader-based protocols use quorum certificates and view changes to obtain safety under Byzantine faults~\cite{castro1999practical,yin2019hotstuff}. FaB established that increasing replication to $n \geq 5f+1$ permits a two-message-delay fast path~\cite{martin2006fast}. Contemporaneously with FaB, Dutta, Guerraoui, and Vukolic characterized the exact relationship among two-round latency, resilience, authentication, and process-role configuration~\cite{dutta2005bestcase}. Recent protocols revisit this design point, including Alpenglow, Kudzu, Hydrangea, and ChonkyBFT~\cite{kniep2025alpenglow,shoup2025kudzu,shrestha2026hydrangea,franca2025chonkybft}. Minimmit is closest to Hermes in structure: with $n \geq 5f+1$, a (2f+1)-vote M-notarization allows processes to advance, while an (n-f)-vote L-notarization finalizes a block~\cite{chou2025minimmit}. When no M-notarization forms, Minimmit instead nullifies the view. Hermes uses the same quorum sizes but changes their semantics: M-notarizations contain pairwise-comparable prefix votes, while an L-notarization may contain $n-f=4f+1$ valid but non-identical votes. Hermes finalizes their heaviest common prefix and carries a supported continuation into the next view. Thus, a view that would be nullified by a single-value protocol can still finalize useful work.

\paragraph{Slow leaders and productive fallback.}
Most leader-based protocols respond to a slow proposer by waiting for a timeout and abandoning the view, creating a tension between aggressive timers and unnecessary leader changes. Ambulance directly targets this problem using protocol-rigged races: a fast leader path races with cooperative recovery work, avoiding a clock-based failure detector~\cite{giridharan2026ambulance}. Hermes takes a complementary approach. It retains an aggressive timer, but all processes disseminate proposals from the start of the view. If the leader path misses the deadline, those proposals become inputs to the fallback vote, and the resulting L-notarization finalizes their common prefix rather than discarding the view. Hermes therefore makes timeout expiration productive without adding a separate recovery protocol.

\paragraph{Prefix consensus.}
Raptr uses prefix voting to tolerate partial data availability within a
leader's ordered proposal~\cite{tonkikh2025raptr}; Hermes instead extracts a
common prefix from independently proposed histories when the leader path
fails. Concurrently, Xiang et al.~\cite{xiang2026prefix} formalizes Prefix
Consensus as a one-shot primitive with low and high outputs, proves tight
three-round bounds at $n=3f+1$, and gives a two-round construction at
$n\geq5f+1$. Their construction counts exact values independently at each
vector coordinate with threshold $n-2f$, making certified prefixes
consistent across quorums. Hermes instead counts extension support over an
arbitrary prefix-ordered space, so every HCP or selected continuation is an
ancestor of an actually voted value. HCPs extracted from different
L-notarizations of one view remain comparable, but their selected
continuations need not be: the $2f+1$ threshold orders candidates only within
one vote set. Hermes therefore does not provide their one-shot high-agreement
guarantee; safety is instead carried across views by justifications, since
every finalizable prefix lies below every valid justification's parent. The shared $5f+1$ resilience follows from
different intersections: threshold against threshold in their construction,
and the $2f+1$ M-notarization against the $4f+1$ fallback witness in Hermes.

\paragraph{Parallel dissemination and multi-proposer ordering.}
DAG-based systems such as Narwhal/Tusk decouple transaction dissemination from consensus so that many processes can produce data concurrently~\cite{danezis2022narwhal}. Autobahn organizes this data into per-process lanes and orders cuts across those lanes~\cite{giridharan2024autobahn}. Hermes adopts a similar lane-based data plane, but represents proposals as parent-relative delta tipcuts with explicit skips and deterministic interleaving. This representation makes distinct fallback proposals comparable as histories and allows their common prefix to be finalized.

Concurrently, Multimmit is the closest contemporary multi-lane design~\cite{lewispye2026multimmit}. It combines Minimmit's one-vote-per-view skeleton with multiple producer chains. Votes are relative to the leader's proposal, report support independently for each chain, and may attest extension blocks beyond the proposed tips. This reduces inclusion latency and confines the effect of a faulty producer largely to its own chain. Hermes differs at the consensus boundary: Multimmit may still nullify a view whose leader block cannot be certified, whereas a Hermes L-notarization need not contain identical votes and directly extracts a finalizable prefix from divergent fallback histories.

\paragraph{Information-theoretic BFT.}
A separate line of work removes transferable signatures and relies only on sender-authenticated point-to-point channels. IT-HS provides an information-theoretic analogue of HotStuff with optimal resilience and bounded storage~\cite{abraham2021iths}. TetraBFT obtains $5\delta$ good-case latency, while Forget-IT and Fast TetraBFT reduce this to the optimal $3\delta$ at $n=3f+1$; Simple-IT provides a simpler $4\delta$ protocol with a $3\delta$ optimistic path~\cite{yu2024tetrabft,abraham2026forgetit,fernandez2026fasttetra,yu2026simpleit}. Most recently, Floor-IT closes the remaining two-round gap: it is an end-to-end signature-free, optimistically responsive protocol with good-case and robust good-case latency $2\delta$ at the tight $n=5f-1$ threshold~\cite{abraham2026floorit}. It additionally proves a worst-case post-GST view length of $2\Delta+2\delta$, $\mathcal{O}(1)$ persistent storage, and $\mathcal{O}(n^2)$ constant-word messages per view.

\section{Conclusion}

Leader-based BFT protocols finalize through their leaders, so the view of a crashed or slow leader is wasted, and no timeout serves both: a conservative one means every crashed leader causes a long time to skip, and an aggressive one may unnecessarily skip the views of leaders that are slightly slow. Hermes makes expired views finalize: votes carry values ordered by a prefix relation, and quorums require comparability rather than equality. Every process broadcasts a justified proposal at view start and casts a single vote; justifications assemble from the votes themselves, so a view has neither nullifications nor view-change rounds. A timely honest leader finalizes its full proposal in $2\delta$, and otherwise any $n-f$ votes, which need not match, finalize their heaviest common prefix at $2\Delta+\delta$. Skipping a slow leader costs the suffix on which the votes diverge, not the view.

Hermes is proved once and instantiated twice. Safety and liveness hold over an abstract prefix value space; the single-lane instantiation is a finality gadget over a dynamically available chain (e.g. Ethereum), with chain tips as values and ancestry as $\preceq$, and the multi-lane instantiation orders Autobahn-style lanes, where deterministic interleaving and parent-relative delta tipcuts with explicit skips keep independent proposals comparable, and sender-indexed erasure-coded amplification keeps every justified value reconstructible by all processes, at $O(n^2\hat{m}+\lambda n^3)$ bits and $O(n^2)$ messages per view, where $\hat{m}$ is the number of lanes. Neither instantiation touches the view structure of Algorithm 1.

\bibliographystyle{plain}
\bibliography{references}
\end{document}